\documentclass[aps,pra,twocolumn,amsfonts,groupedaddress]{revtex4-2}
\usepackage{url}
\usepackage{color}
\usepackage{dcolumn}
\usepackage{bm}
\usepackage{cancel}
\usepackage{graphicx}
\usepackage{comment}
\usepackage{amssymb,amsmath,amsthm}                		                		 
 \newtheorem{thm}{Theorem}
 \newtheorem{lemma}{Lemma}
 \newtheorem{cor}{Corollary}
\DeclareMathOperator{\Tr}{\operatorname{Tr}}
\DeclareMathOperator{\LL}{\operatorname{\mathcal{L}}}
\DeclareMathOperator{\e}{\operatorname{e}}
\begin{document}
  
  \title{Quantum state driving: measurements versus pulses}
  
\author{Yi-Hsiang Chen}
\email{yihchen@isi.edu}
\affiliation{Information Sciences Institute, University of Southern California, Marina del Rey, CA 90292, USA}
\affiliation{Department of Physics and Astronomy, and Center for Quantum Information Science \& Technology,University of Southern California, Los Angeles, California 90089, USA}


\begin{abstract}
The quantum Zeno effect is well-known for fixing a system to an eigenstate by frequent measurements. It is also known that applying frequent unitary pulses induces a Zeno subspace that can also pin the system to an eigenspace. Both approaches have been studied as means to \emph{maintain} a system in a certain subspace. Extending the two concepts, we consider making the measurements/pulses dynamical so that the state can \emph{move} with the motion of the measurement axis/pulse basis. We show that the system stays in the dynamical eigenbasis when the measurements/pulses are slowly changing. Explicit bounds for the apply rate that guarantees a success probability are provided. In addition, both methods are inherently resilient against non-Markovian noise.  Finally, we discuss the similarities and differences between the two methods and their connection to adiabatic quantum computation.
\end{abstract}

\maketitle

 \section{Introduction}
Frequent measurements can fix a state to an eigenstate of the measurement observable. It is known as the quantum Zeno effect (QZE) \cite{Sudarshan1977,Nori2008,PhysRevLett.100.090503,Herrera2014,PhysRevLett.112.070404}. It is often considered an approach of error prevention or suppression for quantum systems \cite{Paz-Silva2012,Vaidman_1996,Chen2020}.  A closely related idea that uses frequent strong unitary pules to average out unwanted noise has been well developed in the area of dynamical decoupling (DD) \cite{DD1998,Viola2002}. Similar to measurements, applying frequent pulses can also create a stabilization for a quantum system. This connection to the QZE is well illustrated in \cite{Facchi2004}, where it shows the frequent pulses induce a Zeno subspace \cite{Facchi2002} which eliminates transitions to other subspaces and provides a freezing effect similar to the QZE. 

  Since frequent measurements pin a state to the measurement basis, one can expect that if the measurement is slowly changing, the state can be steered with it. An experiment of this effect can be found in \cite{Siddiqi2018}, where the authors considered a circuit-QED setting that demonstrates a steering of a superconducting qubit by dynamically changing the measurement operator in a rate slower than the measurement-induced dephasing rate. This has also been analyzed theoretically \cite{AntiZeno2000,Pechen2006,Shuang2008}. In particular, \cite{AntiZeno2000} has shown that a state continuously measured with a time-dependent projector remains in the projected space at all times in the limit of infinite measurement rate. However, a quantitative bound for the required measurement rate that guarantees a given success probability, in a generic setting, is lacking.  On the other hand, frequent unitary pulses induce a Zeno subspace in the unitary's spectrum. If the pulses are slowly changing, the induced subspace will also be time-varying, which drives the system in a similar manner as a slowly changing measurement does. This type of quantum state manipulation via dynamical pulses has not been fully explored and is the other main topic in this paper.
  
  We first present the measurement control method with Theorem \ref{thm1} that explains the setting and result. In contrast to existing works, which often consider only the infinite measurement rate or relatively specific systems, Theorem \ref{thm1} provides an exact bound for the success probability for an arbitrary system that is subjected to an external noise. In particular, we show how high the measurement rate should be to achieve a desired success probability. Following the Zeno steering approach, we demonstrate that the same quantum state driving can also be accomplished by frequent pulses. The result is based on the effect of a fast phase averaging over a function, which can be realized as an extension of the mean ergodic theorem. The mathematical details are provided in Supplemental Material. We found that the required pulse rate is in general higher than the required measurement rate for the same accuracy, although pulses are often considered a less expensive operation than measurements.  We discuss the similarities and differences between the two and argue that both can be realized as an adiabatic quantum computing task \cite{RevModPhys.90.015002}.

  \section{Zeno steering}
 The essence of the quantum Zeno effect can be seen from a simple two-level example. Suppose a state, initially in $|0\rangle$, evolves according to $\exp(-iHt)$ with some Hamiltonian and is measured in the $\{|0\rangle,|1\rangle\}$ basis for every interval $\tau/N$. For each measurement, the probability of getting the outcome $|0\rangle$ given the state was in $|0\rangle$ is $|\langle 0|\exp(-iH\tau/N)|0 \rangle|^2=1-\mathcal{O}(\tau^2/N^2)$, implying the probability that the state never left $|0\rangle$ is $1-\mathcal{O}(\tau^2/N)$, after $N$ measurements over the interval $\tau$. If $N$ is large enough, the state can be fixed at $|0\rangle$ with a probability arbitrarily close to 1. Moreover, this holds true for more general cases where the Hamiltonian jointly acts on the system with its environment. In these cases, the frequent measurements are often considered an error-suppression technique that aims to preserve the system from external disturbance. 
	
	Instead of using Zeno effect to keep the state from evolving, we consider \emph{time-dependent} frequent measurements that drive the state along the measurement axis. Suppose a system $S$ is frequently measured by a sequence of $N$ measurements with projectors $\{\Pi_0,\Pi^{\perp}_0\},$ $\{\Pi_1,\Pi^{\perp}_1\},\dots,$ where each measurement axis is slightly rotated from the previous one, i.e., $\Pi_{j+1}=\e^{iK_j dt}\Pi_j \e^{-iK_jdt}$ with some generator $K_j$. If the system is initially in the projected space of $\Pi_0$, it becomes in the projected space of $\Pi_N$, in the large $N$ limit. A stronger result also holds: the probability that the system always stays in the $\Pi_j$ space approaches 1. If the system is subjected to an interaction with an environment $B$ while being frequently measured, the probability of the system being in the final projected space $\Pi_N$ can still approach 1. The following theorem summarizes these results.

 \begin{thm}\label{thm1}
Suppose a system with Hilbert space $\mathcal{H}_S$ is interacting with an environment with Hilbert space $\mathcal{H}_B$ via a joint Hamiltonian $H_{SB}\in \mathcal{B}(\mathcal{H}_S\otimes\mathcal{H}_B)$.  Suppose we perform a sequence of $N$ projective measurements $\{\Pi_j,\Pi_{j}^{\perp}\}_{j=1}^{N}$ on the system with intervals $\tau/N$, where $\Pi_j=\Pi_j\otimes I_B$ are some projectors acting only on the system. The projector at the $(j+1)$th step is related to the projector at the $j$th step by
\begin{equation}
\Pi_{j+1}=\e^{iK_j \frac{\tau}{N}}\Pi_j \e^{-iK_j \frac{\tau}{N}} 
\end{equation}
with some Hermitian operators $K_j=K_j\otimes I_B$ that rotate the measurement axis. The sequence of measurements is done over a duration $\tau$. If the joint system $\rho\in\mathcal{B}(\mathcal{H}_S\otimes\mathcal{H}_B)$ is initially 
 \begin{equation}
\rho(0)=\Pi_0\rho(0)\Pi_0,
\end{equation} 
then the probability that the final state  $\rho(\tau)$ is in the projected space of $\Pi_N$ after $N$ measurements is
\begin{equation}
P\Big(\rho(\tau)=\Pi_N\rho(\tau)\Pi_N\Big)\geq 1- \epsilon \e^{\epsilon},
\end{equation}
where 
\begin{equation}
\epsilon=\frac{4(K+||H_{SB}||_{\infty})^2\tau^2}{N}\e^{2(K+||H_{SB}||_{\infty})\tau/N}
\end{equation}
and $K\equiv\max_{\forall j}\{||K_j ||_{\infty}\}$ and $||\cdot||_{\infty}$ being the Schatten infinity norm.
\end{thm}
This result partly resembles the Zeno effect in that Hamiltonian dynamics only affect probabilities in the second order of duration $dt$ (=$\tau/N$). In addition, the rotating measurement axis effectively produces another sequence of Hamiltonian driven evolutions for the state. After $N$ measurements, the probability that the system never jumped to an orthogonal part is a product of conditional probabilities that each is $1-\mathcal{O}(dt^2)$. Therefore, the probability of this event happening is $1-\mathcal{O}(\tau^2/N)$, which approaches 1 at large $N$. The exact bound for this success probability requires more careful evaluations and is provided in Supplemental Material. Although the bound holds for any set of parameters, it becomes meaningful only if $\epsilon$ is sufficiently small. One can see that when $\epsilon> 0.57$ the bound becomes negative and no longer provides any usefull information about the success probability. In this regime, it is also possible to experience the anti-Zeno effect \cite{Kofman2000,PhysRevLett.87.040402,KOSHINO2005191,PhysRevB.81.115307,PhysRevA.82.042109}, which states that a moderately small measurement rate can actually increase the decay rate. 

Note that the $H_{SB}$ can in general contain the environment's internal Hamiltonian $H_B$.  However, it is intuitive that $H_B$ should not affect the success probability since $H_B$ commutes with the measurements and $K_j$. Therefore, we have

\begin{cor}\label{cor1}
Follow from Theorem \ref{thm1}, if $H_{SB}$ has components acting purely on the environment $B$, the success probability is independent of those terms, i.e., 
\begin{equation}
P\Big(\rho(\tau)=\Pi_N\rho(\tau)\Pi_N\Big)\geq 1- \epsilon \e^{\epsilon},
\end{equation}
where 
\begin{equation}
\epsilon=\frac{4(K+||H_{SB\setminus B}||_{\infty})^2\tau^2}{N}\e^{2(K+||H_{SB\setminus B}||_{\infty})\tau/N},
\end{equation}
$H_{SB\setminus B}=H_{SB}-H_B$, and $H_B$ has support only on $B$.
\end{cor}
\begin{proof}
First note that $[H_{B},K_j]=[H_B,\Pi_j]=0$ for all $j$. This implies $\LL_{SB}^m\LL_j^n\Pi_j=\LL_{SB\setminus B}^m\LL_j^n\Pi_j$, which leads to the result.
\end{proof}

Theorem \ref{thm1} can represent a control protocol that aims to steer the system to a final target space in $\Pi_N$. During the steering process, the system is affected by its environment with an interaction Hamiltonian $H_{SB}$. The measurement rate is identified as $\lambda=N/\tau$ and 
\begin{equation}
\epsilon=4(K+||H_{SB}||_{\infty})\tau  \frac{K+||H_{SB}||_{\infty}}{\lambda}\e^{2\frac{K+||H_{SB}||_{\infty}}{\lambda}}.
\end{equation}
The successful steering probability is at least $1-\epsilon \e^{\epsilon}$. When the measurement rate is large enough, i.e., $\lambda\gg K+||H_{SB}||_{\infty} $, the success probability is $1-\mathcal{O}((K+||H_{SB}||_{\infty})/\lambda)$. One can realize $K$ as the steering rate of the control protocol, i.e., $K\sim || \dot{\Pi}||_{\infty}$, where $||\dot{\Pi}||_{\infty}$ represents the rate at which the measurement axis changes with time. More specifically, in the continuum, we have 
\begin{align}
&\Pi(t+dt)=\e^{iK(t)dt}\Pi(t)\e^{-iK(t)dt} \nonumber\\
& \implies \ ||\dot{\Pi}||_{\infty}\sim ||K(t) ||_{\infty}.
\end{align}
$||H_{SB}||_{\infty}$ represents the strength of the noise from the interaction between the system and the bath. Corollary \ref{cor1} shows that $H_{SB}$ excludes the environment's internal Hamiltonian (but can include $H_S$ in general).  One can also evaluate the measurement rate that guarantees a required success probability $1-\delta$, i.e., the measurement rate achieving a required success probability $1-\delta$ is
\begin{equation}
\lambda\geq \frac{2(K+||H_{SB}||_{\infty})}{W\left(\frac{W(\delta)}{2(K+||H_{SB}||_{\infty})\tau }\right)},
\end{equation}   
where $W(\cdot)$ is the Lambert W function, which is a monotonically increasing function in the $\mathbb{R}^+$ domain \cite{Corless1996} and $W(x)\to x$ when $x\to0$ and $W(x)\to\log x$ when $x\to\infty$.

The original Zeno effect is recovered by making $K_j=0$ for all $j$, where the system ``remains'' in the projected space $\Pi_0(=\Pi_N)$.

\section{Pulse steering}
It is shown in \cite{Bernad_2017,Facchi2004} that applying frequent unitary pulses induces a Zeno subspace with respect to the unitary's spectrum. In light of this, we consider applying frequent unitary pulses in place of measurements. In particular, we examine whether the steering can be achieved with a weaker operation, i.e., using pulses rather than measurements. We show that the steering can still be achieved up to an arbitrary accuracy with a high enough pulse rate. To gain intuition behind the pulse control method, we first look at the cases where the pulse is changing according to a time-independent generator without any external noise affecting the state. And we extend the result to the most general cases where the pulse is changing according to a time-dependent generator and state is also subjected to an external noise while being steered.

 \emph{Time-independent generator.--- }To gain intuition of this mechanism, we first consider the special cases where the generator of pulse's rotation is time-independent. Suppose the system is applied with a sequence of time-varying unitary pulses between intervals $dt=\tau/N$, i.e., the first pulse is $U_1=\e^{i K dt}U\e^{-i Kdt}$, and the second is $U_2=\e^{i K dt}U_1\e^{-i Kdt}$, etc., where $K$ is a Hermitian operator that generates the rotations of the pulses. The initial unitary $U$ has a spectral decomposition $\sum_{\mu}\e^{-i\phi_{\mu}}\Pi_{\mu}.$

After applying $N$ pulses, the system becomes
\begin{align}
&|\psi(Ndt)\rangle=U_N\cdots U_1 |\psi(0)\rangle \nonumber \\
&=\underbrace{\e^{i K Ndt} U\e^{-i K Ndt} }_{U_N}\underbrace{\e^{i K(N-1) dt} U\e^{-i K (N-1)dt} }_{U_{N-1}} \cdots|\psi(0)\rangle \nonumber\\
&=\e^{iK N dt} \left(U\e^{-iKdt}\right)^{N-1}|\psi(0)\rangle \label{psiN}
\end{align}
It is shown (e.g., Thm. 1 in \cite{Burgarth_2019}) that in the limit of $N\to\infty$,
\begin{equation}
 \left(U\e^{-iKdt}\right)^{N-1}\to U^N\e^{-i\sum_{\mu}\Pi_{\mu}K\Pi_{\mu}\tau}\equiv U_Z, \label{MET}
\end{equation}
where the subscript $Z$ indicates the $U$ induced Zeno subspace. This result is closely related to the Mean Ergodic Theorem \cite{Bernad_2017,reed1980functional}, which says that averaging over all powers of a unitary will eliminate all the subspaces except the one with eigenvalue 1.

 Suppose the initial state $|\psi(0)\rangle$ is in the projected space of $\Pi_{\nu}$, i.e., $\langle\psi(0)|\Pi_{\nu}|\psi(0)\rangle=1$. Using Eq.~(\ref{psiN}) and (\ref{MET}), we see that after $N$ pulses, the ``weight'' in the rotated projected space $\Pi_{\nu,N}$ is 
\begin{align}
&\langle\psi(Ndt)|\Pi_{\nu,N}|\psi(Ndt)\rangle\nonumber \\
&= \langle\psi(Ndt)| \e^{iKNdt}\Pi_{\nu}\e^{-iKNdt}|\psi(Ndt)\rangle \nonumber\\
&\to\langle\psi(0)|U_Z^{\dagger} \Pi_{\nu} U_Z|\psi(0)\rangle= \langle\psi(0)|\Pi_{\nu}|\psi(0)\rangle=1.
\end{align}
This indicates the state is moved to the $\Pi_{\nu}(\tau)=\Pi_{\nu,N}$ projected space after the sequence of pulses, i.e., the state initially in $\Pi_{\nu}$ is dragged to the projected space $\Pi_{\nu}(\tau)$ as the pulse evolves to $U_N=U(\tau)=\sum_{\mu}\e^{-i\phi_{\mu}}\Pi_{\mu}(\tau)$. 

\subsection{Time-dependent generator and noise}
 Now we consider applying a sequence of time-varying pulses where the changes are generated by a time-dependent generator $K(t)$. In addition, we assume the state is also affected by an external noise $H_{SB}$ while we apply the pulses. Suppose the initial unitary $U$ has a spectral decomposition $\sum_{\mu}\e^{-i\phi_{\mu}}\Pi_{\mu}.$ The first pulse is $U_1=\e^{i K(0) dt}U\e^{-i K(0)dt}$, and the second is $U_2=\e^{i K(dt) dt}U_1\e^{-i K(dt)dt}$, etc., where $K(t)$ is a Hermitian operator that generates the rotations of the pulses and is analytic in its real an imaginary part. We denote the state $|\psi\rangle$ the joint state vector that includes the system $S$ and its environment $B$. The pulses $U$ and generator $K$ only act on $S$, and $H_{SB}$ acts on both $S$ and $B$. 

After $N$ applications of the pulses, the state becomes
\begin{align}
|\psi(Ndt)\rangle&=U_N\e^{-iH_{SB}dt}\cdots U_1\e^{-iH_{SB}dt} |\psi(0)\rangle \nonumber \\
&=\prod_{\ell=0}^{N-1}\e^{iK(\ell dt) dt} \prod_{\ell=0}^{N-1}U\e^{-i\widetilde{K}(\ell dt) dt}|\psi(0)\rangle, \label{psiN1}
\end{align}
where $\widetilde{K}(\ell dt)$ is a Hermitian operator satisfying
\begin{align}
&\e^{-i\widetilde{K}(\ell dt) dt}= \\ \label{Ktilde}
&\left( \prod_{\ell'=0}^{\ell-1}\e^{iK(\ell' dt) dt}\right)^{\dagger}\e^{-iK(\ell dt)dt} \e^{-iH_{SB}dt}\left(\prod_{\ell'=0}^{\ell-1}\e^{iK(\ell' dt) dt} \right). \nonumber
\end{align}
The product $\prod$ assumes a descending order in time. Note that $\widetilde{K}(t)$ is another time-dependent Hermitian operator, which in general does not commute with itself at different times. Its analyticity in time is inherited from the exponential products of $K(t)$. Notice the difference between Eq. (\ref{psiN}) and (\ref{psiN1}). In particular, for a time-independent generator we have Eq.~(\ref{MET}), which is a consequence of Mean Ergodic Theorem, while in the time-dependent case we have 
\begin{align}
&\prod_{\ell=0}^{N-1}U\e^{-i\widetilde{K}(\ell \tau/N) \tau/N}\to U^N\mathcal{T}\e^{-i\sum_{\mu}\int_0^{\tau} \Pi_{\mu} \widetilde{K}(s) \Pi_{\mu}ds}\equiv U_Z
\end{align}
in the large $N$ limit. $\mathcal{T}$ is a descending time-ordered operator. This is proved in Identity 1 in Supplemental Material. It uses an extension of Mean Ergodic Theorem, i.e., $\lim_{N\to\infty}(1/N)\sum_{\ell=0}^Nf(\ell/N)\e^{i\ell \phi}=0$ for any analytic function $f$  (see Lemma 2 in S. M.). One can see that the resulting operator $U_Z$ also has a Zeno subspace with respect to $U$'s spectrum. Again, suppose the state $|\psi(0)\rangle$ was initially in the projected space of $\Pi_{\nu}$ (i.e., $\langle\psi(0)| \Pi_{\mu} |\psi(0)\rangle$=1). After the sequence of $N$ pulses the projected space is rotated to $\Pi_{\nu}(\tau)\equiv\Pi_{\nu,N}=\left(\prod_{\ell=0}^{N-1}\e^{iK(\ell dt) dt}\right) \Pi_{\nu}\left( \prod_{\ell=0}^{N-1}\e^{iK(\ell dt) dt}\right)^{\dagger}$. The state after $N$ pulses becomes
\begin{align}
|\psi(Ndt)\rangle\to\prod_{\ell=0}^{N-1}\e^{iK(\ell dt) dt} U_Z|\psi(0)\rangle,
\end{align} 
and (recall $\tau=Ndt$) the weight in the $\Pi_{\nu}(\tau)$ space becomes 
\begin{align}
\langle \psi(\tau)| \Pi_{\nu}(\tau)|\psi(\tau)\rangle&\to\langle \psi(0)| U^{\dagger}_Z\Pi_{\nu}U_Z|\psi(0)\rangle \nonumber\\
&=\langle \psi(0)| \Pi_{\nu}|\psi(0)\rangle=1.
\end{align}
This shows that the system $S$, initially in the $\Pi_{\nu}$ space, is moved to the $\Pi_{\nu}(t)$ space after the operation. This provides a state manipulation method by using the concept of dynamical decoupling, but instead of aiming to decouple from environment, it uses time-varying pulses to \emph{steer} the state to a desired space.

\section{Discussion}
We have shown that a target quantum state can be prepared via frequent measurements or pulses, even under the presence of noise. The measurement state preparation relies on the Zeno effect to pin the state onto the measurement axis and drag it while the measurement axis slowly changes. Similarly, frequent strong pulses induce a Zeno subspace that prohibits transitions to other subspaces. By slowly changing the pulses, the state stays in the time-changing subspace and evolves along the pulses to the target space. In the following, we discuss the similarities and differences between the two. 

In general, applying measurements is considered a more expensive procedure comparing to applying unitaries --- one would need to know the measurement basis to perform a measurement but unitaries pulses can be carried out without the knowledge of the spectrum.  This may make using measurements sounds impractical at first. However, we can consider a scenario that the initial measurement is in some known basis and the measurement evolves according to $\dot{\Pi_l}=i[K(t),\Pi_l],$ where $\Pi_l$ are the projectors of the measurement and $K(t)$ is some generator. Note that an operator $\Pi$ satisfying $\dot{\Pi}=i[K(t),\Pi]$ evolves in the same way as in Theorem \ref{thm1}. Suppose $\{\Pi_l(0)\}$ is some known measurement and the state is initially in the $\Pi_r(0)$ space.  We know that at time $t$ the state will be in the $\Pi_r(t)$ space even though we do not necessarily need to know $\Pi_r(t)$. Similarly, in the unitary pulse setting, the pulse $U$ satisfies $\dot{U}=i[K(t),U]$. A state initially in an eigenspace of $U(0)$ will be in the corresponding eigenspace of $U(t)$ at time $t$. Again, we only know the initial pulse and how it evolves but we do not need to know explicitly what $U(t)$ is. This is reminiscent of a quantum state satisfying a Schrödinger equation --- we know the initial state and the Hamiltonian but we do not in general know what the state is at a particular time. This can have applications to quantum computing in a similar sense of adiabatic quantum computing (AQC) \cite{RevModPhys.90.015002}. In AQC, we initially prepare a state in the ground state of an easy-to-prepare Hamiltonian, and the Hamiltonian is slowly changed to a target Hamiltonian that encodes the problem to be solved. If the rate of change is small enough, the state will be in the target ground state of the evolved Hamiltonian. In a similar sense, the slowly changing Hamiltonian is replaced with a slowly evolving measurement/pulse in the measurement/pulse steering scenario. 

A more technical difference between the measurement and pulse steering is the rate at which we apply the measurements/pulses. To achieve a given accuracy (or success probability), the required pulse rate is in general higher than the required measurement rate. This is because measurements are intrinsically a stronger dephasing operation comparing to unitaries --- each measurement completely eliminates the off-diagonal components while each unitary pulse does not. As a result, measurements provide a stronger fixing onto the basis. This difference is quantitatively reflected in the bounds of Theorem \ref{thm1} and the proof of Identity 1. In the measurement setting, the success probability is lower bounded by a quantity that depends only on the operator norm of the generator and that of the noise term, as shown in Theorem \ref{thm1}. However, in the pulse setting, the success probability depends on not only the norms but also the number of different eigenvalues of the pulse (see proof of Identity 1), which can depend on the system size. Although measurements are more efficient in terms of the apply rate, they are sometimes considered more expensive to carry out in reality, comparing to pulses. Nevertheless, the best use case for both approaches will be casting to an adiabatic evolution of the operator, i.e., $\dot{\Pi}=i[K(t),\Pi]$ or $\dot{U}=i[K(t),U]$, where the initial operation is easy to perform and the evolution of the operation is known and controllable.

\section{Conclusion}
We have provided two methods to manipulate the evolution of a quantum state. The first method is by applying frequent measurements while slowly changing the measurement basis. The state will be fixed onto the basis and follow its movement.  When the measurement rate is much larger than the rate at which the measurement axis changes and the norm of the noise, then the success probability approaches 1. A quantitative bound for the required measurement rate to achieve a success probability is also provided.
The second method is by applying frequent unitary pulses while the pulses are slowly changing. A state initially in an eigenspace of the initial pulse will be steered to the corresponding eigenspace of the final pulse. The weight in the correct final eigenspace approaches 1 as the pulse rate becomes much larger than the pulse change rate, the noise strength and the number of eigenvalues of the pulse. The measurement approach is an extension of quantum Zeno effect while the pulse method is a generalization of the idea of dynamical decoupling. We have also discussed the pros and cons between the two. In particular, the required measurement rate is less than the required pulse rate, but measurements can be harder to carry out comparing to pulses. Finally, we have shown that both methods can be mapped to a state preparation task in a similar sense of the adiabatic quantum computing. A nice bonus of them is that they can intrinsically suppress errors during the operations, in the same way as in the Zeno effect and dynamical decoupling. 

Given the similar structure of adiabatic evolution and measurement/pulse driving, it would be interesting to explore a deeper connection between them and possibly discover a unified framework to describe these phenomena. On the more practical side, an immediate  direction to pursue would be to verify the pulse steering method in experiment, for small systems, as the measurement steering method for a single qubit has been demonstrated in \cite{Siddiqi2018}. With the rapid development in the quantum computing sector, new platforms have been actively proposed and studied. The dynamical measurement/pulse protocols can provide new control capabilities for systems with continuously tunable measurements or pulses.

\begin{acknowledgments}
The author would like to thank Daniel Lidar's mentioning of the connection to dynamical decoupling and valuable comments by Todd Brun and Itay Hen. This work is supported by the U.S. Department of Energy (DOE), Office of Science, Basic Energy Sciences (BES) under Award No. DE-SC0020280.
\end{acknowledgments}

\bibliography{rf}

\widetext
\clearpage
\begin{center}
\textbf{\large Supplemental Material}
\end{center}

\section*{Proof of Theorem 1}
Here, we provide a detailed proof for the Theorem 1 in the main context. We begin by evaluating the probability of success for the first step, and we bound the conditional probabilities between each step. By noticing that the final success probability is lowered bounded by the particular case where we got all outcomes $\Pi$, we obtain the bound given in Theorem 1.

The joint state $\rho$ is initially in the projected space of $\Pi_0$. For the first step, it evolves unitarily for a duration $\tau/N$ and is projectively measured with the measurement $\{\Pi_1,\Pi_1^{\perp}\}$. For simplicity, we denote $dt=\tau/N$. The probability of obtaining the outcome $\Pi_1$ from the first measurement is
\begin{equation}
P(\Pi_1)=\Tr \left[\Pi_1 \e^{\LL_{SB}dt}\rho(0) \Pi_1\right]=\Tr \left[\e^{iK_0dt}\Pi_0 \e^{-iK_0dt} \e^{\LL_{SB}dt}\rho(0)\right]=\Tr\left[\Pi_0 \e^{\LL_0dt}\e^{\LL_{SB}dt}\rho(0)\right],
\end{equation}
where $\LL_{SB}(\cdot)\equiv -i[H_{SB},\ \cdot\ ]$ and $\LL_j(\cdot)\equiv -i[K_j,\ \cdot\ ]$ with $j=0,\cdots,N-1$. Expanding $P(\Pi_1)$ to the first order of $dt$, we observe
\begin{align} 
P(\Pi_1)&=\Tr\left[\Pi_0\rho(0)\right] + \Tr\left[ \Pi_0 \LL_0 \rho(0)\right]dt+\Tr\left[\Pi_0 \LL_{SB}\rho(0) \right]dt+ \mathcal{O}(dt^2) \nonumber\\
&=1-i\underbrace{\Tr[\Pi_0[K_0,\rho(0)]]}_{0}dt-i\underbrace{\Tr[\Pi_0[H_{SB},\rho(0)]]}_{0}dt+\mathcal{O}(dt^2),
\end{align}
where $\Pi_0\rho(0)\Pi_0=\rho(0)$ and $\Tr[A[B,C]]=-\Tr[B[A,C]]$ are used. We see that the probability of $\rho$ being in the projected space of $\Pi_1$ after the first measurement is $1-\mathcal{O}(dt^2)$. 

We can evaluate the probability of getting the outcome $\Pi_{j+1}$ given the previous outcome being $\Pi_j$ as,
\begin{align}
&P(\Pi_{j+1}|\Pi_j)=\Tr \left[\Pi_{j+1} \e^{\LL_{SB}dt}\rho(jdt) \Pi_{j+1}\right]=\Tr \left[\Pi_{j}\e^{\LL_jdt} \e^{\LL_{SB}dt}\rho(jdt)\right] \nonumber\\
&=\Tr\left[\Pi_j \e^{-iK_j dt}\e^{-iH_{SB}dt} \rho(jdt)\e^{iH_{SB}dt}\e^{iK_jdt} \right] \nonumber\\
&=\Tr\left[\rho(jdt)  \e^{iH_{SB}dt}\e^{iK_jdt}\Pi_j \e^{-iK_jdt}\e^{-iH_{SB}dt}\right] \nonumber\\
&=\Tr\left[\rho(jdt) \e^{-\LL_{SB}dt}\e^{-\LL_jdt} \Pi_j \right] \nonumber\\
&=\underbrace{\Tr[\rho(jdt)\Pi_j]}_{1}+i\underbrace{\Tr[\rho(jdt)[K_j,\Pi_j]]}_{0 }dt+i\underbrace{\Tr[\rho(jdt)[H_{SB},\Pi_j]]}_{0}dt +\Tr\left[\rho(jdt)\sum_{n+m\geq2}^{\infty}\frac{(-\LL_{SB}dt)^m}{m!}\frac{(-\LL_{j}dt)^n}{n!}\Pi_j\right] \nonumber\\
&\equiv 1+R_2=1-|R_2|,
\end{align}
where $\Pi_j\rho(jdt)\Pi_j=\rho(jdt)$ is used and we denote 
\begin{equation}
R_2\equiv \Tr\left[\rho(jdt)\sum_{n+m\geq2}^{\infty}\frac{(-\LL_{SB}dt)^m}{m!}\frac{(-\LL_{j}dt)^n}{n!}\Pi_j\right] \nonumber
\end{equation}
the remainder of the expansion starting from the second order of $dt$. Since probability is upper bounded by 1, it must be true that $R_2\leq 0$. $|R_2|$ can be upper bounded by
\begin{align}
|R_2|&=\left|\Tr\left[\rho(jdt)\sum_{n+m\geq2}^{\infty}\frac{(-\LL_{SB}dt)^m}{m!}\frac{(-\LL_{j}dt)^n}{n!}\Pi_j\right]\right| \nonumber\\
&\leq \sum_{n+m\geq2}^{\infty}\left|\Tr\left[\rho(jdt)\frac{(\LL_{SB}dt)^m}{m!}\frac{(\LL_{j}dt)^n}{n!}\Pi_j\right]\right| \nonumber\\
&\leq \sum_{n+m\geq2}^{\infty} \frac{(dt)^{n+m}}{m!n!}\underbrace{\left|\left| \rho(jdt)\right| \right|_{1}}_{1}\left| \left| \LL^m_{SB}\LL^n_j\Pi_j\right|\right|_{\infty},
\end{align}
where $||\cdot||_1$ and $||\cdot||_{\infty}$ are the Schatten 1-norm and $\infty-$norm respectively and the property $|\Tr [AB]|\leq||A||_1||B||_{\infty}$ is used. The term $\left| \left| \LL^m_{SB}\LL^n_j\Pi_j\right|\right|_{\infty}$ can be bounded by
\begin{equation}
\left| \left| \LL^m_{SB}\LL^n_j\Pi_j\right|\right|_{\infty}\leq ||2H_{SB}||^m_{\infty}||2K_j||^n_{\infty}||\Pi_j||_{\infty}\leq ||2H_{SB}||^m_{\infty}(2K)^n,
\end{equation}
where $||\LL(\cdot) ||_{\infty}=||[H,\cdot] ||_{\infty}\leq 2||H||_{\infty}||\cdot||_{\infty}$ is repeatedly used for $\LL^m_{SB}$ and $\LL^n_j$, and $K\equiv \max_{\forall j}||K_j ||_{\infty}$. The remainder becomes
\begin{align}
|R_2|&\leq \sum_{n+m\geq2}^{\infty} \frac{(dt)^{n+m}}{m!n!}||2H_{SB}||^m_{\infty}(2K)^n \nonumber\\
&=\e^{2(||H_{SB}||_{\infty}+K)dt}-1-2(||H_{SB}||_{\infty}+K)dt \nonumber\\
&\leq 4(||H_{SB}||_{\infty}+K)^2dt^2 \e^{2(||H_{SB}||_{\infty}+K)dt}\equiv \frac{\epsilon}{N},
\end{align}
where the final inequality follows from Taylor's remainder theorem for exponential functions and we define 
\begin{equation}
\epsilon\equiv  N4(||H_{SB}||_{\infty}+K)^2 dt^2 \e^{2(||H_{SB}||_{\infty}+K)dt}.
\end{equation}
Therefore, we have 
\begin{equation}
P(\Pi_{j+1}|\Pi_j)=1-|R_2|\geq 1-\frac{\epsilon}{N}.
\end{equation}
We see that the probability of the state being in the space of $\Pi_{j+1}$ given it was in the space of $\Pi_j$ is at least $\epsilon/N$ close to 1. Hence, the probability that we obtain the outcome $\Pi_N,\Pi_{N-1},\cdots,\Pi_1$ for the $N$ projective measurements is
\begin{align}
&P\Big(\text{getting } \Pi_j\ \text{for every measurement } \{\Pi_j,\Pi_j^{\perp}\}\Big) \nonumber\\
&=P\left(\Pi_N|\Pi_{N-1}\right)P\left(\Pi_{N-1}|\Pi_{N-2}\right)\cdots P\left(\Pi_1\right) \nonumber\\
&\geq \left(1-\frac{\epsilon}{N}\right)^N=\sum_{\ell=0}^N \binom{N}{\ell}\left(-\frac{\epsilon}{N}\right)^{\ell}\nonumber \\
&\geq 1-\sum_{\ell=1}^N\binom{N}{\ell}\left(\frac{\epsilon}{N}\right)^{\ell}=1-\sum_{\ell=1}^N \frac{\epsilon^{\ell}}{\ell!}\underbrace{\frac{N}{N}\frac{N-1}{N}\cdots\frac{N-\ell+1}{N}}_{\ell\ \text{terms}} \nonumber\\
&\geq 1-\sum_{\ell=1}^N\frac{\epsilon^{\ell}}{\ell!}=1+1-\sum_{\ell=0}^N\frac{\epsilon^{\ell}}{\ell!}\geq 1+1-\e^{\epsilon}\nonumber \\
&\geq 1-\epsilon \e^{\epsilon},
\end{align}
where the Taylor's remainder theorem for exponentials is used again for the final inequality. From this, we can conclude that the probability of the state being in the $\Pi_N$ space is at least $\epsilon e^{\epsilon}$ close to 1, i.e.,
\begin{equation}
P\Big(\rho(Ndt)=\Pi_N\rho(Ndt)\Pi_N\Big)\geq P\left(\Pi_N|\Pi_{N-1}\right)P\left(\Pi_{N-1}|\Pi_{N-2}\right)\cdots P\left(\Pi_1\right)\geq1- \epsilon \e^{\epsilon}.
\end{equation}
Replacing with $dt=\tau/N$, we obtain
\begin{equation}
P\Big(\rho(\tau)=\Pi_N\rho(\tau)\Pi_N\Big)\geq 1- \epsilon \e^{\epsilon},
\end{equation}
where 
\begin{equation}
\epsilon=\frac{4(K+||H_{SB}||_{\infty})^2\tau^2}{N}\e^{2(K+||H_{SB}||_{\infty})\tau/N}.
\end{equation}
\clearpage

\section*{Mathematical details for the pulse steering}
This section provides the mathematical background for the pulse steering method. The following two Lemmas generalize the fact that the sum of different powers of a phase is of order $O(1)$, i.e., it does not grow with the number of terms in the sum. Theorem 2 is the backbone for the pulse method result, and it generalizes Eq. (16) in \cite{Facchi2004} and Theorem 1 in \cite{Burgarth_2019} to time-dependent Hermitian operators.
\begin{lemma}\label{lm1}
For any real number $\phi\neq 0 $ (mod $2\pi$) and any positive integer $k$,
\begin{equation}
\lim_{N\to\infty}\left|\sum^{N-1}_{n=0}\left(\frac{n}{N}\right)^k \e^{i n \phi} \right|=\frac{1}{\left| 1-\e^{i \phi}\right|}. \nonumber
\end{equation}
\end{lemma}
\begin{proof}
Let us define the finite sequence 
$$
S_N(\phi)\equiv\sum^{N-1}_{n=0} \e^{i n \phi}=\frac{1-\e^{iN\phi}}{1-\e^{i\phi}}.
$$
One can observe that
\begin{align}
\sum^{N-1}_{n=0}\left(\frac{n}{N}\right)^k \e^{i n \phi}=\left(\frac{-i}{N}\right)^k\left(\frac{d}{d\phi}\right)^kS_N(\phi)=\left(\frac{-i}{N}\right)^k\left(\frac{d}{d\phi}\right)^k\frac{1-\e^{iN\phi}}{1-\e^{i\phi}}. \label{phisum1}
\end{align}
As $N$ goes to infinity, we only have to keep track of the $N^k$ order terms since the lower order terms will vanish when multiplying the pre-factor $1/N^k$. One can see that the only term with $O(N^k)$ is the $d^k/d\phi^k$ acting on the phase, i.e.,
$$
\left|\left(\frac{d}{d\phi}\right)^k\frac{1-\e^{iN\phi}}{1-\e^{i\phi}}\right|=\frac{1}{\left|1-\e^{i\phi}\right|}N^k+O\left(N^{k-1}\right).
$$
Combining with the pre-factor $1/N^k$ and taking the limit $N\to\infty$, we obtain the result claimed.
\end{proof}
This indicates that for any polynomial function $p_k(x)$ of order $k$, the infinite series is bounded, i.e., $\left|\sum_{n=0}^{N-1}p_k(n/N)\e^{in\phi}\right|<\infty$ as $N\to\infty$. In addition, the coefficient is constant in $k$. Therefore, one can expect it works for any analytic function $f(x)$, i.e., $\left|\sum_{n=0}^{N-1}f(n/N)\e^{in\phi}\right|<\infty$ as $N\to\infty$.  Intuitively, this is true because of Lemma \ref{lm1} and the fact that any analytic function allows a Taylor expansion from any point in its analytic region, i.e., we can Taylor expand $f(x)$ from the origin,
$$
f(x)=\sum_{k=0}^{\infty}\frac{f^{(k)}(0)}{k!}x^k.
$$
The quantity in question is
\begin{align}
\lim_{N\to\infty}\lim_{K\to\infty}\left|\sum_{n=0}^{N-1}\sum_{k=0}^{K}\frac{f^{(k)}(0)}{k!}\left(\frac{n}{N}\right)^k \e^{i n\phi}\right|&\leq\lim_{K\to\infty}\lim_{N\to\infty}\sum_{k=0}^{K}\frac{\left|f^{(k)}(0)\right|}{k!}\left|\sum_{n=0}^{N-1}\left(\frac{n}{N}\right)^k \e^{i n\phi}\right| \nonumber\\
&=\frac{1}{\left|1-\e^{i \phi}\right|}\sum_{k=0}^{\infty}\frac{\left|f^{(k)}(0)\right|}{k!}<C,
\end{align}
where the first equality uses Lemma \ref{lm1} and the analyticity of $f(x)$. A more formal statement is provided below.

\begin{lemma}\label{lm2}
Let $f(x)$ be a real analytic function in $\mathbb{R}$ and $\phi\neq 0$ (mod $2\pi$). Then there exists a constant $C$ such that
$$
\lim_{N\to\infty}\left|\sum^{N-1}_{n=0}f\left(\frac{n}{N}\right) \e^{i n \phi} \right|<C.
$$
In addition, any partial sum is also bounded by the same constant $C$, i.e.,
$$
\lim_{N\to\infty}\left|\sum^{M}_{n=L}f\left(\frac{n}{N}\right) \e^{i n \phi} \right|<C,
$$
where $M>L$.
\end{lemma}

\begin{proof}

From Taylor's Theorem, we have for any $x$ and any truncation order $K$
$$
f(x)=\sum_{k=0}^{K}\frac{f^{(k)}(0)}{k!}x^k + \frac{f^{(K+1)}(x')}{(K+1)!}x^{K+1},
$$
where $x'\in[0,x]$. 

We first evaluate the full sum. The quantity in question is 
\begin{align}
\left|\sum^{N-1}_{n=0}f\left(\frac{n}{N}\right) \e^{i n \phi}\right|&\leq \left|\sum^{N-1}_{n=0}\sum_{k=0}^{K}\frac{f^{(k)}(0)}{k!}\left(\frac{n}{N}\right)^k \e^{i n \phi}\right| + \left| \sum^{N-1}_{n=0}\frac{f^{(K+1)}(x_n')}{(K+1)!}\left(\frac{n}{N}\right)^{K+1}\e^{in\phi}\right| \nonumber\\
&\equiv S_1+S_2, \label{qnty1}
\end{align}
where $x'_n\in[0,n/N]$. We first evaluate the first term $S_1$, i.e.,
\begin{align}
S_1&=\left|\sum^{N-1}_{n=0}\sum_{k=0}^{K}\frac{f^{(k)}(0)}{k!}\left(\frac{n}{N}\right)^k \e^{i n \phi}\right|\leq \sum_{k=0}^{K}\frac{\left|f^{(k)}(0)\right|}{k!}\left|\sum_{n=0}^{N-1}\left(\frac{n}{N}\right)^k \e^{i n\phi}\right| \nonumber\\
&=\sum_{k=0}^{K}\frac{\left|f^{(k)}(0)\right|}{k!} \frac{1}{N^k}\left|\left(\frac{d}{d\phi}\right)^k\frac{1-\e^{iN\phi}}{1-\e^{i\phi}}\right| \nonumber \\
&=\sum_{k=0}^{K}\frac{\left|f^{(k)}(0)\right|}{k!} \frac{1}{N^k}\left|\sum_{l=0}^{k} \frac{k!}{l!(k-l)!}\left(\frac{d}{d\phi}\right)^l\frac{1}{1-\e^{i\phi}}\left(\frac{d}{d\phi}\right)^{k-l}(1-\e^{iN\phi})\right| \nonumber \\
&\leq \sum_{k=0}^{K}\frac{\left|f^{(k)}(0)\right|}{k!} \sum_{l=0}^{k} \frac{k!}{l!(k-l)!} \left|\frac{1}{N^l}\left(\frac{d}{d\phi}\right)^l\frac{1}{1-\e^{i\phi}}\right| \left|\frac{1}{N^{k-l}} \left(\frac{d}{d\phi}\right)^{k-l}(1-\e^{iN\phi})\right|. \label{sum0}
\end{align}
The last norm is bounded by 2 (i.e., it is 1 for $k-l\geq1$ and $\leq2$ for $k-l=0$). The middle term with $l$th order derivative on $1/(1-\e^{i\phi})$ can be evaluated by observing that there are $2^{l-1}$ terms, where each term is no larger than $l!/N^l C_{\phi}^{l+1}$, where $C_{\phi}=\min\{1,|1-\e^{i\phi}|\}$. Therefore, the middle term can be bounded by $l!/(NC_{\phi}/2)^l2C_{\phi}$. Combing together, we get
\begin{align}
S_1&\leq \frac{1}{C_{\phi}}\sum_{k=0}^{K}\frac{\left|f^{(k)}(0)\right|}{k!} \sum_{l=0}^{k} \frac{k!}{l!(k-l)!} \frac{l!}{(NC_{\phi}/2)^l} \nonumber\\
&=\frac{1}{C_{\phi}}\sum_{k=0}^{K}\frac{\left|f^{(k)}(0)\right|}{k!}  \sum_{l=0}^{k} \frac{k (k-1)\cdots (k-l+1)}{(NC_{\phi}/2)^l} \nonumber\\
&\leq\frac{1}{C_{\phi}}\sum_{k=0}^{K}\frac{\left|f^{(k)}(0)\right|}{k!} \sum_{l=0}^{k}\left(\frac{k}{2K}\right)^l \ \ \ \ \ \text{when}\ \ \ \  N\geq4K/C_{\phi} \nonumber\\
&=\frac{1}{C_{\phi}}\sum_{k=0}^{K}\frac{\left|f^{(k)}(0)\right|}{k!}\frac{1-\left(\frac{k}{2K}\right)^{k+1}}{1-\frac{k}{2K}}\leq \frac{2}{C_{\phi}}\sum_{k=0}^{K}\frac{\left|f^{(k)}(0)\right|}{k!}\to C\ \ \ \ \text{as }\ \ \ K\to\infty, \label{sum1}
\end{align}
where the final quantity converges as $K\to\infty$ (because $f(x)$ is analytic and its Taylor series converges absolutely). 

The second term $S_2$ in (\ref{qnty1}) can be bounded as 
\begin{align}
S_2=\left| \sum^{N-1}_{n=0}\frac{f^{(K+1)}(x_n')}{(K+1)!}\left(\frac{n}{N}\right)^{K+1}\e^{in\phi}\right|\leq \left|\frac{f^{(K+1)}(x^{\star})}{(K+1)!}\right| \left| \sum^{N-1}_{n=0} \left(\frac{n}{N}\right)^{K+1}\right|\leq \left|\frac{f^{(K+1)}(x^{\star})}{(K+1)!}\right| N, \label{S2sum}
\end{align}
where $x^{\star}$ is such that $|f^{(K+1)}(x^{\star})|=\max_{x\in[0,1]}\{|f^{(K+1)}(x)|\}$. Recall that whenever $N\geq4K/C_{\phi}$, $S_1$ is bounded as (\ref{sum1}). We set $N=4K/C_{\phi}$ and we have
\begin{equation}
S_2\leq \left|\frac{f^{(K+1)}(x^{\star})}{(K+1)!}\right| \frac{4K}{C_{\phi}}\leq \frac{4}{C_{\phi}}\left|\frac{f^{(K+1)}(x^{\star})}{(K+1)!}2^{K+1}\right|\to 0\ \ \ \ \text{as } \ \ \ K\to\infty,
\end{equation}
where the last quantity is the $(K+1)$th term of the Taylor expansion of $f(x^{\star}+2)$ from $x^{\star}$, which approaches 0 as $K$ increases. Finally, since the inequality (\ref{qnty1}) holds for any $K$ and $N$, we can set $N=4K/C_{\phi}$ and take $K\to\infty$. Therefore, we obtain the result
$$
\lim_{N\to\infty}\left|\sum^{N-1}_{n=0}f\left(\frac{n}{N}\right) \e^{i n \phi}\right|\leq C.
$$

The partial sum can be evaluated in a similar manner to (\ref{qnty1}), i.e.,
\begin{align}
\left|\sum^{M}_{n=L}f\left(\frac{n}{N}\right) \e^{i n \phi}\right|&\leq \left|\sum^{M}_{n=L}\sum_{k=0}^{K}\frac{f^{(k)}(0)}{k!}\left(\frac{n}{N}\right)^k \e^{i n \phi}\right| + \left| \sum^{M}_{n=L}\frac{f^{(K+1)}(x_n')}{(K+1)!}\left(\frac{n}{N}\right)^{K+1}\e^{in\phi}\right| \nonumber\\
&\equiv S_1+S_2.
\end{align}
The first term is
\begin{align}
S_1&\leq \sum_{k=0}^{K}\frac{\left|f^{(k)}(0)\right|}{k!}\left|\sum_{n=L}^{M}\left(\frac{n}{N}\right)^k \e^{i n\phi}\right|=\sum_{k=0}^{K}\frac{\left|f^{(k)}(0)\right|}{k!} \frac{1}{N^k}\left|\left(\frac{d}{d\phi}\right)^k\e^{iL\phi}\frac{1-\e^{i(M-L)\phi}}{1-\e^{i\phi}}\right|  \nonumber\\
&=\sum_{k=0}^{K}\frac{\left|f^{(k)}(0)\right|}{k!} \frac{1}{N^k}\left|\sum_{l=0}^{k} \frac{k!}{l!(k-l)!}\left(\frac{d}{d\phi}\right)^l\frac{1}{1-\e^{i\phi}}\left(\frac{d}{d\phi}\right)^{k-l}(\e^{iL\phi}-\e^{iM\phi})\right| \nonumber \\
&\leq \sum_{k=0}^{K}\frac{\left|f^{(k)}(0)\right|}{k!} \sum_{l=0}^{k} \frac{k!}{l!(k-l)!} \left|\frac{1}{N^l}\left(\frac{d}{d\phi}\right)^l\frac{1}{1-\e^{i\phi}}\right| \underbrace{\left|\frac{1}{N^{k-l}} \left(\frac{d}{d\phi}\right)^{k-l}(\e^{iL\phi}-\e^{iM\phi})\right|}_{\leq \left(\frac{L}{N}\right)^{k-l}+\left(\frac{M}{N}\right)^{k-l}<2} \\
&\leq\frac{2}{C_{\phi}}\sum_{k=0}^{K}\frac{\left|f^{(k)}(0)\right|}{k!}\to C\ \ \ \ \text{as }\ \ \ K\to\infty, \label{sum3}
\end{align}
where (\ref{sum3}) has the same bound as (\ref{sum1}). The second term is similarly to (\ref{S2sum}) as
\begin{align}
S_2=\left| \sum^{M}_{n=L}\frac{f^{(K+1)}(x_n')}{(K+1)!}\left(\frac{n}{N}\right)^{K+1}\e^{in\phi}\right|\leq  \left|\frac{f^{(K+1)}(x^{\star})}{(K+1)!}\right| (M-L),
\end{align}
which approaches 0 as $K\to\infty$ while $N=4K/C_{\phi}$. This completes the proof.
\end{proof}

\setcounter{thm}{1}
\begin{thm}\label{mainThm}
For a time dependent Hermitian matrix $H(t)$, where the real and imaginary parts of each element are analytic functions of $t$, and a unitary matrix $U$ with a spectral decomposition $\sum_{\mu=1}^m\e^{i\phi_{\mu}}P_{\mu}$, where $\e^{i\phi_{\mu}}\neq \e^{i\phi_{\nu}}$ whenever $\mu\neq\nu$ and $P_{\mu}$ projectors,
\begin{align}
\left|\left| U^{\dagger N}\prod_{\ell=0}^{N-1}U\e^{-iH(\ell t/N) t/N}-\prod_{\ell=0}^{N-1} \e^{-i\sum_{\mu=1}^m\frac{t}{N}P_{\mu}H\left(\frac{\ell t}{N}\right)P_{\mu}}\right|\right|=O\left(\frac{1}{N}\right),
\end{align}
or equivalently 
\begin{align}
\lim_{N\to\infty}U^{\dagger N}\prod_{\ell=0}^{N-1}U\e^{-iH(\ell t/N) t/N}=\mathcal{T}\e^{-i\sum_{\mu=1}^m\int_0^t P_{\mu} H(s) P_{\mu}ds},
\end{align}
where $\mathcal{T}$ is the descending time ordering operator.
\end{thm}
\begin{proof}
First we rewrite the product as 
\begin{align}
&U^{\dagger N}\prod_{\ell=0}^{N-1}U\e^{-iH(\ell t/N) t/N} \nonumber\\
&=U^{\dagger N-1}\e^{-iH((N-1) t/N) t/N} U^{N-1} U^{\dagger N-2}\e^{-iH((N-2) t/N) t/N} U^{N-2}\cdots U^{\dagger}\e^{-iH(t/N) t/N} U\e^{-iH(0)t/N}  \nonumber\\
&=\e^{-iH_{N-1} t/N}\e^{-iH_{N-2}t/N}\cdots \e^{-iH_{0}t/N}=\prod_{\ell=0}^{N-1}\e^{-iH_{\ell}t/N},
\end{align}
where
$$
H_{\ell}\equiv U^{\dagger \ell}H(\ell t/N)U^{\ell }.
$$
The product of exponentials can be Taylor expanded and compared at the same order of $t/N$, i.e.,
\begin{align}
\prod_{\ell=0}^{N-1}\e^{-iH_{\ell}t/N} &=\sum_{k_0,\cdots,k_{N-1}=0}^{\infty}\left(\frac{-it}{N}\right)^{k_{N-1}}\cdots\left(\frac{-it}{N}\right)^{k_0}\frac{H^{k_{N-1}}_{N-1}}{k_{N-1}!}\cdots\frac{H^{k_0}_0}{k_0!} \nonumber\\
&=\sum_{k=0}^{\infty}\frac{1}{k!}\left(\frac{-it}{N}\right)^k\sum_{k_0+\cdots+k_{N-1}=k}\frac{k!}{k_0!\cdots k_{N-1}!}H^{k_{N-1}}_{N-1}\cdots H^{k_0}_0\nonumber\\
&=\sum_{k=0}^{\infty}\frac{1}{k!}\left(\frac{-it}{N}\right)^k \mathcal{T}\left[\left(\sum_{\ell=0}^{N-1}H_{\ell}\right)^k\right] \label{multinomial1}\\
&=I+\sum_{k=1}^{\infty}\frac{1}{k!}\left(\frac{-it}{N}\right)^k\sum_{\ell_1,\cdots,\ell_{k}=0}^{N-1}\mathcal{T}(H_{\ell_1}\cdots H_{\ell_k}), \label{multinomial2}
\end{align}
where $\mathcal{T}$ is the descending ordering operator, i.e, $\mathcal{T}(H_{j_1}\cdots H_{j_{k}})=H_{r_1}\cdots H_{r_k}$ with $r_1\geq \cdots \geq r_k$ and $\{r_q\}_{q=1}^k=\{j_q\}_{q=1}^k$, and (\ref{multinomial1}) is from the multinomial theorem. Recall that the spectral decomposition of $U$ is $\sum_{\mu=1}^{m}\e^{i\phi_{\mu}}P_{\mu}$. Now we insert the spectral projectors $P_{\mu}$ and separate (\ref{multinomial2}) into two sums, i.e.,
\begin{align}
&\prod_{\ell=0}^{N-1}\e^{-iH_{\ell}t/N} \nonumber\\
&=I+\sum_{k=1}^{\infty}\frac{1}{k!}\left(\frac{-it}{N}\right)^k\sum_{\ell_1,\cdots,\ell_{k}=0}^{N-1} \left[\sum_{\mu=1}^m \mathcal{T}(P_{\mu}H_{\ell_1}P_{\mu}\cdots P_{\mu}H_{\ell_k}P_{\mu})+ \sum_{(\mu_0,\cdots,\mu_{k})\in \sigma}  \mathcal{T}(P_{\mu_0}H_{\ell_1}P_{\mu_1}\cdots P_{\mu_{k-1}}H_{\ell_k}P_{\mu_k}) \right], \label{term1}
\end{align}
where $\mathcal{T}$ only order $\ell_j$ but not $\mu_j$, and
$$
\sigma=\{\mu_0,\cdots,\mu_{k}\in\{1,\cdots,m\} \ | \ \text{excluding any case with } \mu_0=\mu_1\cdots=\mu_{k}\}.
$$ 
 Recall from the definition of $H_{\ell}$, we have $P_{\mu}H_{\ell}P_{\mu}=P_{\mu}H(t\ell/ N)P_{\mu}$. Inserting this equality back to the first sum in (\ref{term1}) and deriving backwards along (\ref{multinomial2}) and (\ref{multinomial1}) by replacing each $H_{\ell}$ with $P_{\mu}H(t\ell/ N)P_{\mu}$, we have 
\begin{align}
\prod_{\ell=0}^{N-1}\e^{-iH_{\ell}t/N}&=\sum_{\mu=1}^m P_{\mu}\prod_{\ell=0}^{N-1} \e^{-i\frac{t}{N}P_{\mu}H\left(\frac{\ell t}{N}\right)P_{\mu}} P_{\mu}\nonumber\\
&\ \ \ + \sum_{k=1}^{\infty}\frac{1}{k!}\left(\frac{-it}{N}\right)^k \underbrace{\sum_{\ell_1,\cdots,\ell_{k}=0}^{N-1}  \sum_{(\mu_0,\cdots,\mu_{k})\in \sigma}  \mathcal{T}(P_{\mu_0}H_{\ell_1}P_{\mu_1}\cdots P_{\mu_{k-1}}H_{\ell_k}P_{\mu_k})}_{\equiv R_k}, \label{term2}
\end{align}
where the product $\prod$ assumes descending order. One can see that the first term is a unitary operator block-diagonal in $U$'s spectrum $\{P_{\mu}\}$. The second term will decay to 0 as $N$ goes to infinity. This is mainly due to Lemma \ref{lm2} (applied to the real and imaginary part of each element of $H$). For each $\mu\neq \nu$ (implying $\e^{i\phi_{\mu}}\neq \e^{i\phi_{\nu}}$), we have
\begin{align}
 \left|\left|P_{\mu}\left(\sum_{\ell=0}^{N-1}H_{\ell}\right)P_{\nu}\right|\right|=\left|\left| P_{\mu}\left[\sum_{\ell=0}^{N-1}\e^{i\ell(\phi_{\nu}-\phi_{\mu})}H\left(t\ell /N\right)\right]P_{\nu}\right|\right| \leq \left|\left| \sum_{\ell=0}^{N-1}\e^{i\ell(\phi_{\nu}-\phi_{\mu})}H\left(t\ell /N\right)\right|\right| \leq C \label{bound1}
\end{align}
for some constant $C$. For convenience, we use operator norm (or max norm) for $||\cdot||$. In particular, it says that there exists a constant $C$ that bounds $\left|\left|P_{\mu}\left(\sum_{\ell=0}^{N-1}H_{\ell}\right)P_{\nu}\right|\right|$ for any $N$. In general, $C$ depends on $\mu$ and $\nu$. However, we can choose $C$ to be the largest among all $\mu$, $\nu$.   

To evaluate $R_k$, we first break the sum with the ordered products into different sets of restricted region, i.e., for each set of $(\mu_0,\cdots,\mu_{k})\in \sigma$,
\begin{align}
&\sum_{\ell_1,\cdots,\ell_{k}=0}^{N-1}\mathcal{T}(P_{\mu_0}H_{\ell_1}P_{\mu_1}\cdots P_{\mu_{k-1}}H_{\ell_k}P_{\mu_k}) \nonumber\\
&=k!\sum_{(\ell_1,\cdots,\ell_k)\in\Theta_k}P_{\mu_0}H_{\ell_1}P_{\mu_1}\cdots P_{\mu_{k-1}}H_{\ell_k}P_{\mu_k} \nonumber \\
&\ \ \ +\frac{k!}{2!}\sum_{(\ell_1,\cdots,\ell_k)\in\Theta_{k-1}}P_{\mu_0}H_{\ell_1}P_{\mu_1}\cdots P_{\mu_{k-1}}H_{\ell_k}P_{\mu_k} \nonumber\\
&\ \ \ +\frac{k!}{3!}\sum_{(\ell_1,\cdots,\ell_k)\in\Theta_{k-2}}P_{\mu_0}H_{\ell_1}P_{\mu_1}\cdots P_{\mu_{k-1}}H_{\ell_k}P_{\mu_k} \nonumber \\
&\ \ \ \ \ \ \ \ \ \ \ \ \ \ \ \ \ \ \ \ \ \ \ \ \ \ \ \ \vdots \ \ \ \ \ \ \ \ \ \ \ \ \ \ \ \ \ \ \ \ \ \ \ \ \ \ \ \ \nonumber \\
&\ \ \ +\frac{k!}{k!}\sum_{(\ell_1,\cdots,\ell_k)\in\Theta_{1}}P_{\mu_0}H_{\ell_1}P_{\mu_1}\cdots P_{\mu_{k-1}}H_{\ell_k}P_{\mu_k},
\end{align}
where 
\begin{align}
&\Theta_{k}=\{\ell_1,\cdots,\ell_k\in\{0,1,\cdots,N-1\}\ |\ \ell_1>\ell_2>\cdots> \ell_k\} \nonumber\\
&\Theta_{k-1}=\{\ell_1,\cdots,\ell_k\in\{0,1,\cdots,N-1\}\ |\ \ell_1>\cdots> \ell_k,\  \text{where 1 ``$>$''  is replaced by ``=''}\} \nonumber\\
&\Theta_{k-2}=\{\ell_1,\cdots,\ell_k\in\{0,1,\cdots,N-1\}\ |\ \ell_1>\cdots> \ell_k,\ \text{where 2 ``$>$'' are replaced by ``=''}\} \nonumber\\
&\ \ \ \ \ \ \ \ \ \ \ \ \ \ \ \ \ \ \ \ \ \ \ \ \ \ \ \ \vdots \ \ \ \ \ \ \ \ \ \ \ \ \ \ \ \ \ \ \ \ \ \ \ \ \ \ \ \ \nonumber \\
&\Theta_{1}=\{\ell_1,\cdots,\ell_k\in\{0,1,\cdots,N-1\}\ |\ \ell_1=\ell_2\cdots= \ell_k\}. 
\end{align}
For example, let $k=4$, then $(5=5>4>2)\in\Theta_{4-1}$, $(6>5=5>2)\in\Theta_{4-1}$, $(15>10>3=3)\in\Theta_{4-1}$, etc.. The prefactor $k!/q!$ means there are $k!/q!$ ways to allocate $k$ numbers to $k$ slots with $q$ of them being the same. Note that
the number of elements in each set is $|\Theta_k|={{N}\choose{k}}$, $|\Theta_{k-1}|={{N}\choose{k-1}}$, etc..
Therefore, we have 
\begin{align}
||R_k||&=\left|\left|\sum_{(\mu_0,\cdots,\mu_{k})\in \sigma}\sum_{\ell_1,\cdots,\ell_{k}=0}^{N-1}\mathcal{T}(P_{\mu_0}H_{\ell_1}P_{\mu_1}\cdots P_{\mu_{k-1}}H_{\ell_k}P_{\mu_k}) \right|\right| \nonumber \\
&\leq \sum_{(\mu_0,\cdots,\mu_{k})\in \sigma}\Bigg(\left|\left|k!\sum_{(\ell_1,\cdots,\ell_k)\in\Theta_k}P_{\mu_0}H_{\ell_1}P_{\mu_1}\cdots P_{\mu_{k-1}}H_{\ell_k}P_{\mu_k}\right|\right| \nonumber \\
&\ \ \ \ \ \ \ \ \ \ \ \  \ \ \ \ \ \ \ \ \  +\left|\left|\frac{k!}{2!}\sum_{(\ell_1,\cdots,\ell_k)\in\Theta_{k-1}}P_{\mu_0}H_{\ell_1}P_{\mu_1}\cdots P_{\mu_{k-1}}H_{\ell_k}P_{\mu_k}\right|\right| \nonumber \\
&\ \ \ \ \ \ \ \ \ \ \ \ \ \ \ \ \ \ \ \ \ \ \ \ \ \ \ \ \ \ \ \ \ \ \ \ \ \ \ \ \ \ \ \ \ \ \ \ \ \ \ \ \vdots \ \ \ \ \ \ \ \ \ \ \ \ \ \ \ \ \ \ \ \ \ \ \ \ \ \ \ \ \nonumber \\
&\ \ \ \ \ \ \ \ \ \ \ \  \ \ \ \ \ \ \ \ \  +\left|\left|\frac{k!}{k!}\sum_{(\ell_1,\cdots,\ell_k)\in\Theta_{1}}P_{\mu_0}H_{\ell_1}P_{\mu_1}\cdots P_{\mu_{k-1}}H_{\ell_k}P_{\mu_k}\right|\right|\Bigg) \nonumber \\
&\leq \sum_{(\mu_0,\cdots,\mu_{k})\in \sigma}\Bigg(\left|\left|k!\sum_{(\ell_1,\cdots,\ell_k)\in\Theta_k}P_{\mu_0}H_{\ell_1}P_{\mu_1}\cdots P_{\mu_{k-1}}H_{\ell_k}P_{\mu_k}\right|\right| \nonumber \\
&\ \ \ \ \ \ \ \ \ \ \ \  \ \ \ \ \ \ \ \ \  +\frac{k!}{2!}\sum_{(\ell_1,\cdots,\ell_k)\in\Theta_{k-1}}||P_{\mu_0}||||H_{\ell_1}||||P_{\mu_1}||\cdots ||P_{\mu_{k-1}}||||H_{\ell_k}||||P_{\mu_k}|| \nonumber \\
&\ \ \ \ \ \ \ \ \ \ \ \ \ \ \ \ \ \ \ \ \ \ \ \ \ \ \ \ \ \ \ \ \ \ \ \ \ \ \ \ \ \ \ \ \ \ \ \ \ \ \ \ \vdots \ \ \ \ \ \ \ \ \ \ \ \ \ \ \ \ \ \ \ \ \ \ \ \ \ \ \ \ \nonumber \\
&\ \ \ \ \ \ \ \ \ \ \ \  \ \ \ \ \ \ \ \ \  +\frac{k!}{k!}\sum_{(\ell_1,\cdots,\ell_k)\in\Theta_{1}}||P_{\mu_0}||||H_{\ell_1}||||P_{\mu_1}||\cdots ||P_{\mu_{k-1}}||||H_{\ell_k}||||P_{\mu_k}|| \Bigg) \nonumber \\
&\leq\sum_{(\mu_0,\cdots,\mu_{k})\in \sigma}\Bigg(\left|\left|k!\sum_{(\ell_1,\cdots,\ell_k)\in\Theta_k}P_{\mu_0}H_{\ell_1}P_{\mu_1}\cdots P_{\mu_{k-1}}H_{\ell_k}P_{\mu_k}\right|\right| \nonumber \\
&\ \ \ \ \ \ \ \ \ \ \ \  \ \ \ \ \ \ \ \ \  +\underbrace{\frac{k!}{2!} {{N}\choose{k-1}}}_{<{{k}\choose{1}}N^{k-1}} H_{\max}^k+\underbrace{\frac{k!}{3!} {{N}\choose{k-2}} }_{<{{k}\choose{2}}N^{k-2}}H_{\max}^k+\cdots+ \underbrace{\frac{k!}{k!} {{N}\choose{1}}}_{<{{k}\choose{k-1}}N} H_{\max}^k\Bigg) \nonumber \\
&\leq\sum_{(\mu_0,\cdots,\mu_{k})\in \sigma}\Bigg(\left|\left|k!\sum_{(\ell_1,\cdots,\ell_k)\in\Theta_k}P_{\mu_0}H_{\ell_1}P_{\mu_1}\cdots P_{\mu_{k-1}}H_{\ell_k}P_{\mu_k}\right|\right|+2^kN^{k-1}H^{k}_{\max}\Bigg). \label{term0}
\end{align}
The remaining term is also $O(N^{k-1})$ due to Lemma \ref{lm2}. This can be shown by first observing that for any $(\mu_0,\cdots,\mu_{k})\in \sigma$, there exists a pair $\mu_{j-1}\neq \mu_{j}$. Therefore we can use Lemma \ref{lm2} to factor out a constant for the $\ell_j$ term, i.e.,
\begin{align}
&\left|\left|k!\sum_{(\ell_1,\cdots,\ell_k)\in\Theta_k}P_{\mu_0}H_{\ell_1}P_{\mu_1}\cdots P_{\mu_{k-1}}H_{\ell_k}P_{\mu_k}\right|\right| \nonumber\\
&\leq k!  \sum_{(\ell_1,\cdots, \ell_{j-1},\ell_{j+1},\cdots\ell_k)\in\Theta^{(\ell_j)}_k}\left|\left|\sum_{\ell_j\geq\ell_{j+1}}^{\ell_j\leq\ell_{j-1}}P_{\mu_0}H_{\ell_1}P_{\mu_1}\cdots P_{\mu_{k-1}}H_{\ell_k}P_{\mu_k}\right|\right|  \nonumber\\
&\leq k!  \sum_{(\ell_1,\cdots, \ell_{j-1},\ell_{j+1},\cdots\ell_k)\in\Theta^{(\ell_j)}_k} \left|\left|P_{\mu_{j-1}}\left(\sum_{\ell_j\geq\ell_{j+1}}^{\ell_j\leq\ell_{j-1}}H_{\ell_j}\right)P_{\mu_j}\right|\right| H^{k-1}_{\max} \nonumber\\
&\leq k! C H^{k-1}_{\max}\sum_{(\ell_1,\cdots, \ell_{j-1},\ell_{j+1},\cdots\ell_k)\in\Theta^{(\ell_j)}_k}= k! C H^{k-1}_{\max} {{N}\choose{k-1}}\leq k N^{k-1}C H^{k-1}_{\max}, \label{term3}
\end{align}
where $\Theta_{k}^{(\ell_j)}$ indicates $\Theta_{k}$ removing the entry $\ell_j$. Note that (\ref{term3}) uses Lemma \ref{lm2}, which says there exists a constant that bounds \emph{any} partial sum. Inserting (\ref{term3}) back to (\ref{term0}), we get
\begin{align}
||R_k||&\leq(k N^{k-1}C H^{k-1}_{\max} + 2^kN^{k-1}H^{k}_{\max})\left(\sum_{(\mu_0,\cdots,\mu_{k})\in \sigma}\right) \nonumber\\
&\leq  N^{k-1} H^{k}_{\max}2^k\underbrace{\left(\frac{C}{H_{\max}}+1\right)}_{\equiv C'} \left(\sum_{\mu_0,\cdots,\mu_{k}=0}^m \right)= C' N^{k-1} H^{k}_{\max}2^k m^{k+1}. \label{Rbound}
\end{align}
Combining (\ref{Rbound}) with (\ref{term2}), we have
\begin{align}
&\left|\left|\prod_{\ell=0}^{N-1}\e^{-iH_{\ell}t/N}-\sum_{\mu=1}^m P_{\mu}\prod_{\ell=0}^{N-1} \e^{-i\frac{t}{N}P_{\mu}H\left(\frac{\ell t}{N}\right)P_{\mu}} P_{\mu} \right|\right|=\left|\left| U^{\dagger N}\prod_{\ell=0}^{N-1}U\e^{-iH(\ell t/N) t/N}-\prod_{\ell=0}^{N-1} \e^{-i\sum_{\mu=1}^m\frac{t}{N}P_{\mu}H\left(\frac{\ell t}{N}\right)P_{\mu}}\right|\right| \nonumber\\
&=\left|\left| \sum_{k=1}^{\infty}\frac{1}{k!}\left(\frac{-it}{N}\right)^k R_k\right|\right|\leq \sum_{k=1}^{\infty}\frac{1}{k!}\left(\frac{t}{N}\right)^k ||R_k|| \nonumber\\
&\leq \sum_{k=1}^{\infty}\frac{1}{k!}\left(\frac{t}{N}\right)^k C' N^{k-1} H^{k}_{\max}2^k m^{k+1} \nonumber \\
&=\frac{C'm}{N}\sum_{k=1}^{\infty}\frac{(2mtH_{\max})^k}{k!}\leq\frac{C'm}{N}\e^{2mtH_{\max}}=O\left(\frac{1}{N}\right).
\end{align} 

\end{proof}

\end{document}